\documentclass{article}

\usepackage{PRIMEarxiv}

\usepackage[utf8]{inputenc} 
\usepackage[T1]{fontenc}    
\usepackage{hyperref}       
\usepackage{url}            
\usepackage{booktabs}       
\usepackage{amsfonts}       
\usepackage{nicefrac}       
\usepackage{microtype}      
\usepackage{lipsum}
\usepackage{fancyhdr}       
\usepackage{graphicx}       

\usepackage{cite}
\usepackage{amsmath,amssymb,mathrsfs,amsthm}
\usepackage[ruled,linesnumbered]{algorithm2e}

\usepackage{textcomp}
\usepackage{xcolor}
\usepackage{float}
\usepackage{subfigure}
\usepackage{cleveref,verbatim}
\usepackage{comment}
\usepackage{indentfirst}

\graphicspath{{media/}}     

\newtheorem{theorem}{Theorem}

\def\BibTeX{{\rm B\kern-.05em{\sc i\kern-.025em b}\kern-.08em
		T\kern-.1667em\lower.7ex\hbox{E}\kern-.125emX}}

\title{Value-based Proactive Caching for Sensing Data in Vehicular Networks: An Operator's Perspective
}

\author{
  Yantong Wang, Ke Liu, Hui Ji, Jiande Sun \\
  School of Information Science and Engineering \\
  Shandong Normal University \\
   Ji'nan, 250358, China\\
  \texttt{\{yantong, liuke,~hui.ji,~jiandesun\}@sdnu.edu.cn} \\
}

\begin{document}
\maketitle

\begin{abstract}
Access to sensing data (SD) is crucial for vehicular networks to ensure safe and efficient transportation services. Given the vast volume of data involved, proactive caching required SD is a pivotal strategy for alleviating network congestion and improving data accessibility. 
Despite merits, existing studies predominantly address SD caching within a single slot. Therefore, these approaches lack scalability for scenarios involving multi-slots and are not well-suited for network operators who manage resources within a long-term cost budget. Moreover, the oversight of service capacity at caching nodes may result in substantial queuing delays for SD reception.
To tackle these limitations, we jointly consider the problem of anchoring SD caching and allocating from an operator's perspective. A value model incorporating both temporal and spacial characteristics is given to estimate the significance of various caching decisions. Subsequently, a stochastic programming model is proposed to optimize the long-term system performance, which is converted into a series of online optimization problem by leveraging the Lyapunov method and linearized via introducing auxiliary variables. To expedite the solution, we provide a binary quantum particle swarm optimization based algorithm with quadratic time complexity. 
Numerical investigations demonstrate the superiority of proposed algorithms compared with other schemes in terms of energy consumption, response latency, and cache-hit ratio.
\end{abstract}


\section{Introduction}
\label{sec:introduction}

Vehicular networks integrating vehicles, roadways, and cloud infrastructure to enable seamless communication and data exchange have emerged as a pivotal technology for smart transportation.  
As a cornerstone of the ecosystem, sensing data (SD) are generated by various sensors on vehicles and road side units (RSUs), which improves the environmental perception for mobile vehicles and supports key applications such as traffic management, navigation, and collision avoidance.
However, the continuous growth in vehicular connections and the high demand for real-time SD have resulted in significant pressure on both the fronthaul and hackhaul network~\cite{chen2023tasks}. To alleviate this, proactive caching strategies have been widely adopted to bring content closer to users, thereby reducing latency and improving data accessibility. 

Compared with infotainment content, SD is time-sensitive and only valid within a specific range and time slots\cite{li2023joint}. Therefore, the caching strategies designed for entertainment contents, either in a single-slot\cite{yang2023dynamic} or a long-term fashion\cite{choi2024intelligent}, cannot be applied to the caching design for SD directly. The research focusing on SD dissemination has mainly made caching decisions among RSUs and base stations (BSs) based on an utility model which captures characteristics such as temporally-variant feature\cite{liu2021cooperative}, spacial validity\cite{chai2019hierarchical}, dynamic popularity\cite{khan2024proactive} and hybrid revenue metrics\cite{wu2023hybrid}. Additionally, caching policies for both infotainment content and SD are jointly considered in\cite{zhang2020hierarchical,khodaparas2024intelligent}.

However, the aforementioned studies on SD caching have yet to account for the service capability in terms of queuing, which potentially escalate response latency. 
Moreover, the majority of existing works typically make SD caching decisions in a single time slot, which hinders scalability in multi-slots scenarios, especially considering long-term resource constraints coupling caching decisions over time. Generally, network operators manage resources within the context of a long-term cost budget rather than a short-term one. As a result, these approaches are not suitable for network operators.

To fill the aforementioned research gaps and offer caching decisions tailored to the needs of network operators, we shift the SD caching in vehicular networks to a multi-slots approach and consider the request allocation for load balancing purpose. 
A value model is proposed by jointly considering the temporal and spacial attributes of SD, encompassing factors such as freshness, spatial validity, traffic regulations, and popularity. 
Then we develop a long-term stochastic integer nonlinear programming model aimed at maximizing the long-term value of SD caching, equated to the potential earnings of network operators, while adhering to constraints related to energy budgets, delay tolerance, and resource limitations.
Departing from previous studies that concentrate on vehicle-level decisions, our approach adopts a traffic region-centric viewpoint. This shift is justified by the network operator's priority, which is to serve the collective needs within the defined regions. 
To this end, the main contributions are summarized as follows: 
\begin{itemize}
	\item A long-term region-centric programming model addressing both caching placement and request allocation is formulated, based on a value model considering temporal and spacial characteristics of SD;
	\item  We decompose the model via Lyapunov optimization and then linearize it through auxiliary variables. Moreover, a heuristic algorithm with quadratic complexity is implemented to generate solutions recursively;
	\item To understand the performance gap among evaluated methods, a wide set of simulations is conducted under varying degrees of traffic congestion.
\end{itemize}

The rest of the paper is organized as follows. Section \ref{sec:model} presents the mathematical programming model for SD caching in IoVs. In Section \ref{sec:algorithm}, we discuss the decomposition of stochastic model via Lyapunov optimization and propose a BQPSO-based heuristic algorithm to obtain suboptimal solutions. Section \ref{sec:simulation} and \ref{sec:conclusions} provide the performance evaluation and conclusions respectively.
\section{System Model}
\label{sec:model}
In this section, we first introduce the network model for communications. Then the value model is presented to measure the significance of different decisions. Moreover, the energy cost and delay model are depicted to estimate the constraints of various request allocations. To this end, we give the optimization model considering long-term time-averaged performance. For ease of reference, the main notations are summarized in Table\ref{tab:notation}.

\begin{table}[htbp]
 \caption{Summary of Main Notations}
  \centering
  \begin{tabular}{cl}
    \toprule
    \textbf{Symbol} & \textbf{Description} \\
    \midrule
    $\{0\}$ & Base Station (BS)\\
    $\mathcal{I}=\{1,2,\cdots,I\}$ & Set of road side units (RSUs)\\
    $\mathcal{J}=\{1,2,\cdots,J\}$ & Set of regions\\
    $\mathcal{K}=\{1,2,\cdots,K\}$ & Set of sensing data (SD)\\
    $\mathcal{T}=\{0,1,\cdots,T-1\}$ & Set of time slots\\
    $\mathcal{I}_j\subseteq\mathcal{I}$ & Set of RSUs covering region $j$\\
    $\mathcal{J}_i\subseteq\mathcal{J}$ & Set of regions under coverage of RSU $i$\\
    \midrule
    $r_{ij}(t)$ & Transmission rate between RSU $i$ and region $j$ during slot $t$\\
    $r_{0j}(t)$ & Transmission rate between BS and region $j$ during slot $t$\\
    $B_i$ & Channel bandwidth of RSU $i$\\
    $h_{ij}(t)$ & Channel gain between RSU $i$ and region $j$ during slot $t$\\
    $P_i$ & Transmission power of RSU $i$\\
    $P_0$ & Transmission power of BS\\
    $N_0$ & Noise power\\
    $F_{k}(t)$ & Freshness ratio of SD $k$ during slot $t$\\
    $A_{ik}(t)$ & Binary variable indicating whether RSU $i$ is affected by SD $k$ during $t$ or not\\
    $H_{jk}(t)$ & Popularity of SD $k$ to region $j$ during slot $t$\\
    $r_{jk}(t)$ & Interests interval between last two SD $k$ requests before slot $t$\\
    $z_{jk}(t)$ & Time gap between last SD $k$ request and current frame $t$\\
    $I_{jk}(t)$ & Accumulated requests number for SD $k$ from region $j$ before $t$\\
    $w$ & Power efficiency for caching 1-bit data\\
    $\tau$ & Time slot length\\
    $s_k$ & Size of SD $k$\\
    $S_i$ & Caching capacity of RSU $i$\\
    $\bar{E}$ & Time-averaged energy budget\\
    $d_{jk}(t)$ & Demand intensity for SD $k$ in region $j$ during slot $t$\\
    $\mu_i$ & Request service capacity of RSU $i$\\
    $\delta_j$ & Delay tolerance of region $j$ \\
    \midrule
    $x_{ik}(t)$ & Decision variable indicating whether SD $k$ is cached in RSU $i$ during $t$\\
    $y_{ijk}(t)$ & Decision variable representing the number of requests for SD $k$ from region $j$ \\
    &is served by RSU $i$ during $t$\\
    \bottomrule
  \end{tabular}
  \label{tab:notation}
\end{table}

\subsection{Network Model}
\begin{figure}[htbp]
  \centering
  \fbox{\includegraphics[width=.8\textwidth]{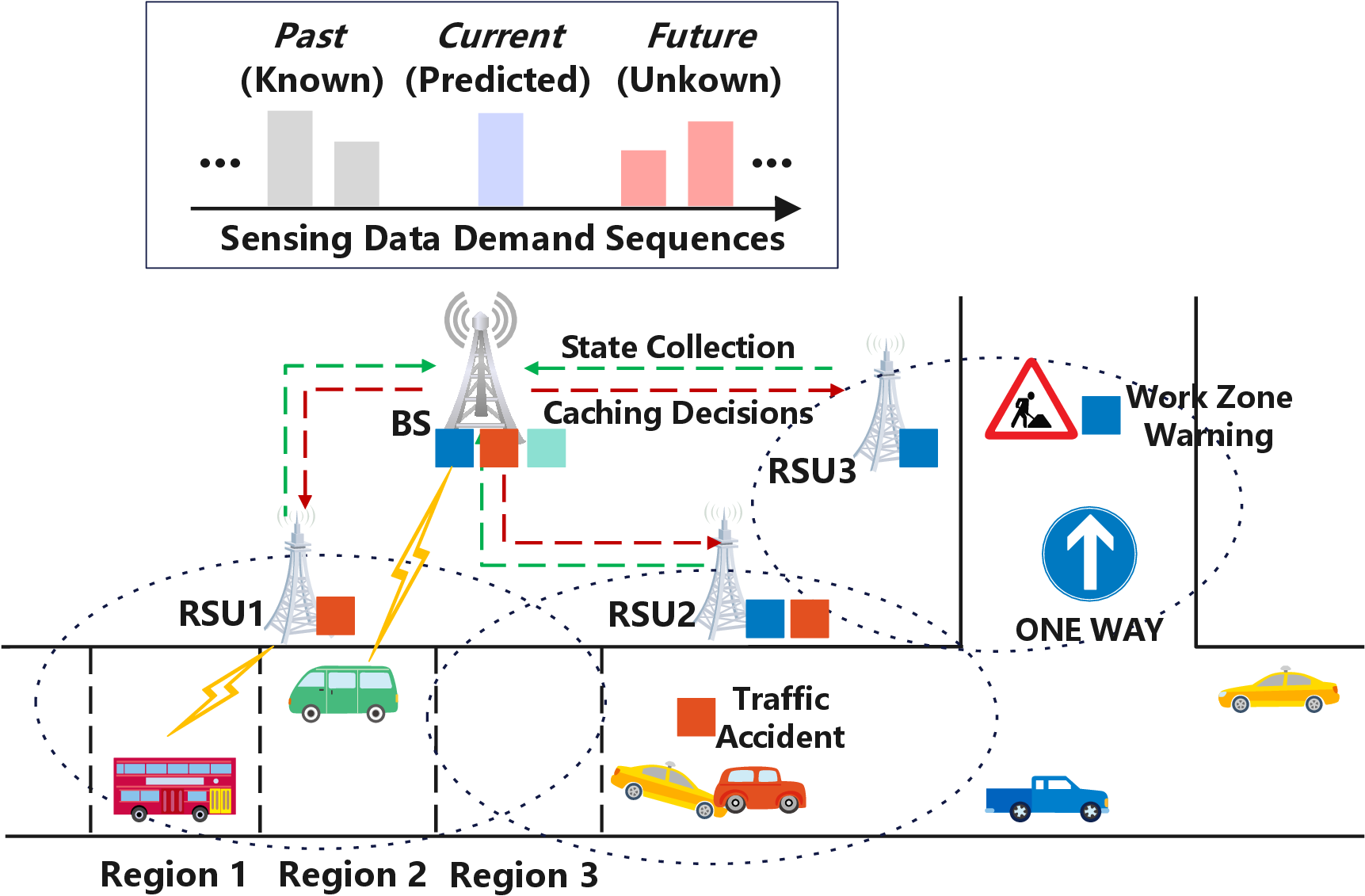}}
  \caption{Proactive caching in a traffic status monitoring scenario.}
  \label{fig:structure}
\end{figure}

As shown in Fig.\ref{fig:structure}, we consider a typical vehicular edge network in urban ereas consisting of one base station (BS) and $I$ road side units (RSUs), indexed by $\{0\}$ and $\mathcal{I}=\{1,2,\cdots,I\}$ respectively. Both BS and RSU are equipped with storage facilities and thus can provide caching services to vehicles under coverage. We assume each RSU has a limited storage capacity $S_i$ and the BS has a sufficiently large space to host all unexpired contents. The coverage of BS is divided into $J$ disjoint regions, indexed by $\mathcal{J}=\{1,2,\cdots,J\}$. Due to the dense deployment of RSUs, the vehicle in region $j$ is in the transmission range of multiple RSUs, denoted by $\mathcal{I}_j\subseteq\mathcal{I}$. The set of regions that are covered by RSU $i$ is represented by $\mathcal{J}_i\subseteq\mathcal{J}$. 
In addition, RSUs and vehicles are equipped with various sensors (camera, Radar, GPS, etc.) to generate a series of SD $\mathcal{K}=\{1,2,\cdots,K\}$ regrading traffic status.

The caching system is assumed to operate in a time-slotted manner with frames $t\in\mathcal{T}=\{0,1,\cdots,T-1\}$, which matches the duration of SD replacement. For the purpose of mathematical tractability, we suppose the popularity distribution and vehicle mobility between regions change slowly with respect to the replacement\footnote{The performance of the proposed method might deteriorate when violating this assumption.}. To achieve centralized control, each RSU is required to send previously recorded SD requests and network states to BS at the beginning of each frame. Then BS makes caching decisions accordingly. If the requested SD is cached at nearby RSUs, vehicles can be served instantly. Otherwise this SD should be retrieved from the BS. 

Basically, the transmission rate depends on the distance between vehicle and the connected RSU\cite{liu2021cooperative}. By making each region small and applying OFDMA as transmission mode\cite{qin2020collaborative}, the data rate between RSU $i$ and region $j$ becomes 
\begin{equation}
\label{fml:rate}
    r_{ij}(t)=\mathbb{E}\left[B_i\log\left(1+\frac{h_{ij}(t)P_i}{N_0}\right)\right],
\end{equation}  
where $B_i$ is the predefined channel bandwidth of RSU $i$, $h_{ij}(t)$ is the channel gain between RSU $i$ and region $j$ during frame $t$, $P_i$ is the transmission power of RSU $i$ and $N_0$ is the noise power. Since the channel condition is not realized at the decision time, we use the expected transmission rate in \eqref{fml:rate}. Let $r_{0j}$ be the transmission data rate between BS and region $j$. Considering BS can only provide limited ratio resources from transmission\cite{chen2020collaborative}, we have $r_{oj}<r_{ij},\forall i\in\mathcal{I}_j$, which is the motivation on edge caching.

In this paper, the decisions of SD caching and request allocation are considered jointly. Therefore, the following decision variable sets are introduced:
\begin{itemize}
	\item \textbf{SD Caching}: $\mathbf{X}\!=\!\{x_{ik}(t)|x_{ik}(t)\in\{0,1\},\! i\in\mathcal{I},\!k\in\mathcal{K},\!t\in\mathcal{T}\}$, where $x_{ik}=1$ indicates that SD $k$ is cached in RSU $i$, otherwise, we have $x_{ik}=0$.
	\item \textbf{Request Allocation}: $\mathbf{Y}\!=\!\{y_{ijk}(t)|y_{ijk}(t)\in\mathbb{N},\! i\in\mathcal{I},\!j\in\mathcal{J},\!k\in\mathcal{K},\!t\in\mathcal{T}\}$ represents the number of request for SD $k$ from region $j$ that is served by RSU $i$.
\end{itemize}

It should be noted that in this paper we consider the decision for region $j$ instead of individual vehicle, as our emphasis is on the long-term expected performance of the vehicular caching system from the network operator's perspective, and the region captures statistical characteristics of SD requests. 

\subsection{Value Model}
Considering the temporal and spacial characteristics of SD, we employ the freshness ratio $F_k(t)$, the affected scope indicator $A_{ik}(t)$ and the popularity $H_{jk}(t)$ to estimate the value of caching a specific SD in an RSU. The freshness ratio can be measured by the remaining time slots before expiration as $$F_k(t)=\max\left\{\frac{(t^E_k-t)}{(t^E_k-t^U_k)},0\right\},$$
where $t^U_k$ and $t^E_k$ are the last updating time slot and the expiration frame of SD $k$ respectively. Small $f_k(t)$ values show that SD $k$ is stale.  

$A_{ik}(t)$ is a binary indicator, where $A_{ik}(t)=1$ means that RSU $i$ locates in the affected scope of SD $k$ during time slot $t$, and $A_{ik}(t)=0$, otherwise. $A_{ik}(t)$ is affected by the space validity of SD and traffic control rules. 
A toy example is shown in Fig.\ref{fig:structure}. Considering the road covered by RSU3 is a one-way street in time slot $t=1$, the traffic accident sensing data $k=1$ happened in RSU2 is not attractive to the vehicle in RSU $i=3$, which results in $A_{31}(1)=0$. Meanwhile, the vehicles in RSU $i=2$ are within the space validity range of the work zone warning data $k=2$ which is generated in RSU3. Therefore we have $A_{22}(1)=1$. 

Furthermore, a caching operation should correspond to the SD request pattern, which is based on the content popularity. Then, we employ the idea of popularity value from \cite{khan2024proactive} and define $H_{jk}(t)$ as follows:
$$H_{jk}(t)=\frac{1}{2}\left[\alpha\cdot e^{-\frac{r_{jk}(t)}{I_{jk}(t)}}+\beta\cdot e^{-\frac{z_{jk}(t)}{I_{jk}(t)}}\right],$$
where $\alpha$ and $\beta$ are impact factors to control the weight; $r_{jk}(t)$ is the requesting interval between last two requests for SD $k$ from region $j$ before current slot $t$; $z_{jk}(t)$ is the time gap between last SD $k$ requesting slots and current frame; and $I_{jk}(t)$ is the accumulated requests number for SD $k$ from region $j$ before frame $t$. Therefore, a larger $H_{jk}(t)$ indicates SD $k$ is frequently-requested by region $j$. To this end, the value model is built as
\begin{equation}
V(t)=\sum_{i\in\mathcal{I}}\sum_{j\in\mathcal{J}_i}\sum_{k\in\mathcal{K}}A_{ik}(t)F_k(t)H_{jk}(t)\cdot y_{ijk}(t).
\end{equation} 

\subsection{Cost Model}

From the perspective of network operator, SD caching results in additional expenditures such as the energy consumption for data storage and delivery.
In this paper, the caching energy is modeled as a proportional function of content size $$E_i^C(t)=w\sum_{k\in\mathcal{K}}s_k\cdot x_{ik}(t)\tau,$$ where $w$ is the power efficiency for storing 1-bit data; $\tau$ is the time slot length; and $s_k$ is the size of SD $k$. As aforementioned, the BS caches all unexpired SDs, which is independent to the SD caching decision $\mathbf{X}$. Therefore the caching cost in BS is viewed as a constant and is omitted in the following analysis for simplification.  Additionally, the transmission energy consumption of RSU $i$ can be given by $$E_i^T(t)=\sum_{j\in\mathcal{J}_i}\sum_{k\in\mathcal{K}}P_is_kr^{-1}_{ij}(t)y_{ijk}(t).$$ And the transmission cost for BS becomes $$E_0^T(t)=\sum_{j\in\mathcal{J}}\sum_{k\in\mathcal{K}}P_0s_kr^{-1}_{0j}(t)[d_{jk}(t)-\sum_{i\in\mathcal{I}_j}y_{ijk}(t)],$$
where $d_{jk}(t)$ is the request intensity for SD $k$ in region $j$ during slot $t$. As illustrated in Fig.\ref{fig:structure}, the BS does not know what the future SD requests are and when they will arrive. 
After that, the total energy cost in frame $t$ can be represented by $$E(t)=\sum_{i\in\mathcal{I}}[E_i^C(t)+E_i^T(t)]+E_0^T(t).$$ Furthermore, the caching service provider generally orchestrates the network within a long-term cost budget. Let $\bar{E}$ denote the time-averaged cost budget. The expected long-term energy consumption constraint can be expressed by
\begin{equation*}
	C_1:~\lim_{T\rightarrow\infty}\frac{1}{T}\sum_{t=0}^{T-1}\mathbb{E}\left[E(t)\right]\leq\bar{E}.
\end{equation*}

\subsection{Delay Model}

Generally, the delay in communication systems can be divided into processing delay, transmission delay, propagation delay, and queuing delay. Here different caching and requests assignments would pose less impact to the propagation delay. Hereafter we focus on the other three kinds of latency. For the transmission delay between served RSU $i$ and region $j$, we have $$T_{ij}(t)={\sum_{k\in\mathcal{K}}s_k\cdot y_{ijk}(t)}{r^{-1}_{ij}(t)}.$$
In the case of cache missing, the transmission delay from BS $\{0\}$ to region $i$ becomes
$$
	T_{0j}(t)={\sum_{k\in\mathcal{K}}s_k\left[d_{jk}(t)-\sum_{i\in\mathcal{I}_j}y_{ijk}(t)\right]}{r^{-1}_{0j}(t)}.
$$

According to \cite{wang2019proactive}, the access to caching service at RSU $i$ follows an M/M/1 model, in which the processing time for each SD request is an exponential distribution with average $\mu_i$. Moreover, the arriving rate in frame $t$ is the number of allocated request. Thus the expected sojourn latency, which includes processing delay and queuing delay, is expressed as
$$
	L_i(t)=\left[\mu_i-\sum_{j\in\mathcal{J}_i}\sum_{k\in\mathcal{K}}y_{ijk}(t)\right]^{-1}.
$$
Similarly, the sojourn latency in BS for cache missing case can be written as
$$
	L_0(t)=\left\{\mu_0-\sum_{j\in\mathcal{J}}\sum_{k\in\mathcal{K}}\left[d_{jk}(t)-\sum_{i\in\mathcal{I}_j}y_{ijk}(t)\right]\right\}^{-1}.
$$

Notice that requests from region $j$ can be served by multiple RSUs $\mathcal{J}_i$ in parallel. Therefore the response latency for region $j$ is represented by 
$$
	D_j(t)=\max\{L_i(t)+T_{ij}(t),L_0(t)+T_{0j}(t)\}.
$$
To guarantee the quality of caching service, the expected response latency for region $j$ cannot exceed the correlated delay tolerance $\delta_j$:
\begin{equation*}
	C_2:~\mathbb{E}\left[D_j(t)\right]\leq\delta_j, \quad \forall j\in\mathcal{J},t\in\mathcal{T}.
\end{equation*}

\subsection{Optimization Model}

For each RSU, the utilized caching space should be limited within the storage capacity $S_i$:
\begin{equation*}
	C_3:~\sum_{k\in\mathcal{K}}s_k\cdot x_{ik}(t)\leq S_i,\quad \forall i\in\mathcal{I},t\in\mathcal{T}.
\end{equation*}

The request allocation $\mathbf{Y}$ is influenced by the caching operation $\mathbf{X}$ and the demand intensity $d_{jk}(t)$, which means
\begin{equation*}
	C_4:~y_{ijk}(t)\leq M\cdot x_{ik}(t),\quad \forall \! i\in\mathcal{I},\!j\in\mathcal{J},\!k\in\mathcal{K},\!t\in\mathcal{T},
\end{equation*}
\begin{equation*}
	C_5:~\sum_{i\in\mathcal{I}_j}y_{ijk}(t)\leq d_{jk}(t),\quad \forall \!j\in\mathcal{J},\!k\in\mathcal{K},\!t\in\mathcal{T},
\end{equation*}
where $M$ is a sufficiently large number. $C_4$ guarantees the RSU providing caching service holds relevant SD. $C_5$ indicates upper-bound of the number of request allocations. 

To obtain the optimal decisions, we maximize the expected long-term time-averaged caching value under the constraints of cost budget, delay tolerance, and resource limitation, which is formulated as the following model:
\begin{equation}
	\begin{aligned}
		(\mathbf{P1}):\quad&\mathop{\max}_{\mathbf{X},\mathbf{Y}}\lim_{T\rightarrow\infty}\frac{1}{T}\sum_{t=0}^{T-1}\mathbb{E}\left[V(t)\right]\\
		\textrm{s.t.}\quad&C_1-C_5.\notag
	\end{aligned}
\end{equation}
Theoretically, $(\mathbf{P1})$ is a stochastic integer nonlinear programming model with the dynamic network status. As a result, it is impractical to solve $(\mathbf{P1})$ without global information over the long run, which is difficult to predict in advance. Furthermore, the caching decisions are coupled with long-term energy constraints in $C_1$, which indicates that consuming more energy in some slots will squeeze the available energy for future usage. Therefore, it is challenging to solve $(\mathbf{P1})$.  
\section{Online Caching Decision Algorithm}
\label{sec:algorithm}
\subsection{Problem Transformation}
Recently, Lyapunov optimization has been viewed as an efficient method to decouple long-term problems. Following this idea, we convert the original problem into a series of online optimization problems. At the beginning, we construct a virtual queue $Q(t)$ to measure the exceeded energy of current caching decision at the end of time slot $t$, which is defined as
\begin{equation}
	\label{efml:virtual_queue}
	Q(t+1)\triangleq\max\{Q(t)+E(t)-\bar{E},0\}.
\end{equation}
The initial queue backlog $Q(0)$ is $0$. 
The following theorem provides an equivalent expression of the long-term budget constraint:
\begin{theorem}
\label{the1}
	Satisfying constraint $C_1$ turns into guaranteeing the stability of virtual queue.
\end{theorem}
\begin{proof}
\label{pro1}
	According to the definition \eqref{efml:virtual_queue}, we have $$Q(t+1)-Q(t)\geq E(t)-\bar{E}.$$
	By summing over $t\in\{0,1,\cdots,T-1\}$ and dividing $T$ on both sides, the above inequality becomes $$\frac{Q(T)}{T}\geq \frac{1}{T}\sum_{t=0}^{T-1}E(t)-\bar{E}.$$
	Take expectation and set limit as $T$ approaches infinity, which results in
	$$\lim_{T\rightarrow\infty}\frac{\mathbb{E}\left[Q(T)\right]}{T}\geq \lim_{T\rightarrow\infty}\frac{1}{T}\sum_{t=0}^{T-1}\mathbb{E}\left[E(t)\right]-\bar{E}.$$
	Since virtual queue is stable, i.e. $$C_6: \lim_{T\rightarrow\infty}\frac{\mathbb{E}\left[Q(T)\right]}{T}=0,$$ constraint $C_1$ is satisfied. 
\end{proof}

Then the \textit{Lyapunov function} is defined as 
$$
	L[Q(t)]\triangleq Q(t)^2/2, 
$$
which presents the congestion level of energy consumption queue. Moreover, we introduce a one-slot \textit{Lyapunov drift} to push $L[Q(t)]$ towards a lower value, which is
$$
	\Delta[Q(t)]\triangleq\mathbb{E}\left\{L[Q(t+1)]-L[Q(t)]|Q(t)\right\}.
$$
The theorem below gives the boundary of \textit{Lyapunov drift}:
\begin{theorem}
\label{the2}
	$$\Delta[Q(t)]\leq\mathcal{B}+Q(t)\mathbb{E}[E(t)-\bar{E}|Q(t)],$$ where $$\mathcal{B}=(E^2_{\text{max}}+\bar{E}^2)/2$$ is a constant number and $$E_{\text{max}}=\max_{t\in\mathcal{T}}E(t).$$
\end{theorem}
\begin{proof}
\label{pro2}
    According to definition \eqref{efml:virtual_queue}, the following statement holds: $$Q(t+1)^2\leq[Q(t)+E(t)-\bar{E}]^2,$$
    which is equivalent to $$Q(t+1)^2\leq Q(t)^2+\left[E(t)-\bar{E}\right]^2+2Q(t)\left[E(t)-\bar{E}\right].$$
    Take conditional expectation on both sides and then we have
    \begin{align}
	\Delta[Q(t)]&=\frac{1}{2}\mathbb{E}\left\{Q(t+1)^2-Q(t)^2|Q(t)\right\}\leq\frac{1}{2}\mathbb{E}\{[E(t)-\bar{E}]^2+2Q(t)\left[E(t)-\bar{E}\right]|Q(t)\}\notag\\
	&\leq\frac{1}{2}\mathbb{E}\{E(t)^2+\bar{E}^2+2Q(t)[E(t)-\bar{E}]|Q(t)\}\leq\mathcal{B}+Q(t)\mathbb{E}[E(t)-\bar{E}|Q(t)],\notag
    \end{align}
    which is the upper boundary of \textit{Lyapunov drift}.
\end{proof}

After that, the \textit{Lyapunov drift-plus-penalty} function in each slot can be written as
\begin{align}
\label{fml:ineq1}
	\Delta[Q(t)]-\mathcal{V}\cdot\mathbb{E}[V(t)|Q(t)]\leq\mathcal{B}+Q(t)\mathbb{E}[E(t)-\bar{E}|Q(t)]-\mathcal{V}\cdot\mathbb{E}[V(t)|Q(t)].
\end{align} 
The weight $\mathcal{V}$ is utilized to adjust the trade-off between energy cost minimization and caching value maximization. Notice that $Q(t)\bar{E}$ is a constant at every frame $t$ and does not affect the caching decision-making. To this end, the original stochastic optimization problem $\mathbf{P1}$ can be transformed to the following single-time-slot problem $\mathbf{P2}$. 
\begin{equation}
 	\begin{aligned}
 		(\mathbf{P2}):\quad&\mathop{\min}_{\mathbf{X}(t),\mathbf{Y}(t)}[Q(t)E(t)-\mathcal{V}\cdot V(t)]\\
 		\textrm{s.t.}\quad&C_2-C_5.\notag
 	\end{aligned}
\end{equation}

\begin{algorithm}[!htbp]
	\caption{Online Caching Decision Algorithm (OCDA)}
	\label{alg:online}
	\KwIn{$s_k$,$S_i$,$\mu_i$,$B_i$,$P_i$,$\delta_j$,$\bar{E}$}
	\KwOut{$\mathbf{X}$, $\mathbf{Y}$}
	Initialize $Q(0)\leftarrow0$\;
	\ForEach{$t=0,1,\cdots,T-1$}{
		Predict SD demand $d_{jk}(t)$\;
		Observe $A_{ik}(t)$, $h_{ij}(t)$\;
		Solve the problem $\mathbf{P2}$ to get $\{\mathbf{X}(t),\mathbf{Y}(t)\}$\;
		Update $Q(t+1)$\; 
	}
\end{algorithm}
The online decision algorithm is shown in Alg.\ref{alg:online}. It is worth noting that though the global information over long run is difficult to predict, the short-term prediction for the immediate time slot is available and many accurate enough prediction algorithms have been developed\cite{li2016trend}. Furthermore, the virtual queue stability constraint $C_8$ can be safely omitted due to the following theorem:
\begin{theorem}
\label{the3}
    The solution of $\mathbf{P2}$ satisfies constraint $C_6$.
\end{theorem}
\begin{proof}
\label{pro3}
    We consider a observation-only policy for $\mathbf{P2}$ that makes caching decision $[\mathbf{X}(t),\mathbf{Y}(t)]$ via a predefined probability $\mathcal{P}$ given observation $[d_{jk}(t),A_{ik}(t),h_{ij}(t)]$, i.e.
    $$[\mathbf{X}^\dag(t),\mathbf{Y}^\dag(t)]=\arg\max\mathcal{P}\left[\mathbf{X}(t),\mathbf{Y}(t)|d_{jk}(t),A_{ik}(t),h_{ij}(t)\right].$$
    Let $E^\dag(t)$ and $V^\dag(t)$ denote the corresponding energy consumption and obtained value based on the observation-only policy, respectively. We suppose the solution $[\mathbf{X}^\dag(t),~\mathbf{Y}^\dag(t)]$ satisfies
    \begin{equation}
    \label{fml:assume1}
        \exists\varepsilon>0,\mathbb{E}\left[E^\dag(t)-\bar{E}\right]\leq-\varepsilon.
    \end{equation}
    Based on the inequality \eqref{fml:ineq1}, we have
    \begin{align}
        \Delta[Q(t)]-\mathcal{V}\cdot\mathbb{E}[V(t)|Q(t)]&\leq\mathcal{B}+Q(t)\mathbb{E}[E(t)-\bar{E}|Q(t)]-\mathcal{V}\cdot\mathbb{E}[V(t)|Q(t)]\notag\\
        &\leq\mathcal{B}+Q(t)\mathbb{E}[E^\dag(t)-\bar{E}|Q(t)]-\mathcal{V}\cdot\mathbb{E}[V^\dag(t)|Q(t)]\label{fml:ineq2}\\
        &=\mathcal{B}+Q(t)\mathbb{E}[E^\dag(t)-\bar{E}]-\mathcal{V}\cdot\mathbb{E}[V^\dag(t)]\label{fml:eq1}\\
        &\leq\mathcal{B}-\varepsilon\cdot Q(t)-\mathcal{V}\cdot\mathbb{E}[V^\dag(t)].\label{fml:ineq3}
    \end{align}
    The inequality \eqref{fml:ineq2} is because the \textit{drift-plus-penalty} provides the optimal solution of $\mathbf{P2}$ and thus overcomes decisions derived from the observation-only policy. The equality \eqref{fml:eq1} holds since $\mathbf{X}^\dag(t)$ and $\mathbf{Y}^\dag(t)$ are only determined by observation $[d_{jk}(t),A_{ik}(t),h_{ij}(t)]$ and thus independent of $Q(t)$. The inequality \eqref{fml:ineq3} is because the assumption \eqref{fml:assume1}.

    Considering $V(t)=\sum_{i}\sum_{j}\sum_{k}A_{ik}(t)F_k(t)H_{jk}(t)\cdot y_{ijk}(t)$, where $A_{ik}(t)\in\{0,1\}$, $F_k(t),H_{jk}(t)\in[0,1]$ and $y_{ijk}(t)\leq d_{jk}(t)$, we can find the boundary of $V(t)$, i.e. $$V_{min}\leq V(t)\leq V_{max},\forall t\in\mathcal{T},$$ where $V_{min}=0$ and $V_{max}=\max_{t\in\mathcal{T}}\sum_j\sum_kd_{jk}(t)$ are constant. Then the inequality \eqref{fml:ineq3} can be rewritten as
    \begin{equation}
        \Delta[Q(t)]-\mathcal{V}\cdot V_{max}\leq\mathcal{B}-\varepsilon\cdot Q(t)-\mathcal{V}\cdot V_{min}.\notag
    \end{equation}
    Take expectation on both side, we have
    \begin{equation}
        \mathbb{E}\left[\Delta\left(Q(t)\right)\right]=\frac{1}{2}\mathbb{E}\left[Q(t+1)^2\right]-\mathbb{E}\left[Q(t)^2\right]\leq\mathcal{B}+\mathcal{V}(V_{max}-V_{min})-\varepsilon\cdot\mathbb{E}\left[Q(t)\right].\notag
    \end{equation}
    By summing over $t\in\{0,1,\cdots,T-1\}$, the above inequality becomes
    \begin{equation}
    \label{fml:ineq7}
        \mathbb{E}\left[Q(T)^2\right]\leq 2T\mathcal{B}+2T\mathcal{V}(V_{max}-V_{min})-2\varepsilon\cdot\sum_{t=0}^{T-1}\mathbb{E}\left[Q(t)\right]\leq2T\mathcal{B}+2T\mathcal{V}(V_{max}-V_{min}).
    \end{equation}
    According to Cauchy-Schwarz Inequality that $\left[\mathbb{E}(XY)\right]^2\leq\mathbb{E}(X^2)\mathbb{E}(Y^2)$, we have
    \begin{equation}
        \left[\mathbb{E}(Q(T))\right]^2\leq\mathbb{E}\left[Q(T)^2\right]\leq2T\mathcal{B}+2T\mathcal{V}(V_{max}-V_{min}).\notag
    \end{equation}
    Then we can obtain the following inequality
    \begin{equation}
        \lim_{T\rightarrow\infty}\frac{\mathbb{E}\left[Q(T)\right]}{T}\leq\lim_{T\rightarrow\infty}\sqrt{\frac{2\mathcal{B}+2\mathcal{V}(V_{max}-V_{min})}{T}}=0.\notag
    \end{equation}
    Therefore, constraint $C_6$ is guaranteed.
\end{proof}

The objective function of $\mathbf{P2}$ represents the supreme bound of $\mathbf{P1}$. By tuning $\mathcal{V}$, the online decision algorithm makes a $[O(1/\mathcal{V}),O(\mathcal{V})]$ trade-off between caching value $V(t)$ and energy consumption $E(t)$, which is presented in \textbf{Theorem~\ref{the4}}. Therefore, a larger $\mathcal{V}$ will achieve a good caching performance, with the cost of virtual queue $Q(t)$ instability. 

\begin{theorem}
\label{the4}
    For any non-negative $\mathcal{V}$, the solution of $\mathbf{P2}$ satisfies that
    \begin{align}
    \label{fml:ineq4}
        \lim_{T\rightarrow\infty}\frac{1}{T}\sum_{t=0}^{T-1}\mathbb{E}\left[V^\star(t)-V(t)\right]\leq\frac{\mathcal{B}}{\mathcal{V}};\\
    \label{fml:ineq5}
        \lim_{T\rightarrow\infty}\frac{1}{T}\sum_{t=0}^{T-1}\mathbb{E}\left[Q(t)\right]\leq\frac{\mathcal{B}+\mathcal{V}(V_{max}-V_{min})}{\varepsilon},
    \end{align}
    where $V^\star(t)$ is the value obtained by solving $\mathbf{P1}$; $V(t)$ and $Q(t)$ are the value and queue's backlog of $\mathbf{P2}$.
\end{theorem}
\begin{proof}
\label{pro4}
    Based on the observation-only policy introduced in aforementioned \textit{Proof} of \textbf{Theorem}~\ref{the3}, we further give an optimal-observation-only policy $[\mathbf{X}^\star(t),\mathbf{Y}^\star(t)]$ which is the optimal solution of $\mathbf{P1}$\cite{neely2010stochastic}.  Let $E^\star(t)$ and $V^\star(t)$ denote the corresponding energy consumption and obtained value based on the observation-only policy, respectively. Moreover, we assume $\mathbb{E}[E^\star(t)-\bar{E}]\leq0$.

    According to the inequality \eqref{fml:ineq1}, we have
    \begin{align}
        \Delta[Q(t)]-\mathcal{V}\cdot\mathbb{E}[V(t)|Q(t)]&\leq\mathcal{B}+Q(t)\mathbb{E}[E(t)-\bar{E}|Q(t)]-\mathcal{V}\cdot\mathbb{E}[V(t)|Q(t)]\notag\\
        &\leq\mathcal{B}+Q(t)\mathbb{E}[E^\star(t)-\bar{E}|Q(t)]-\mathcal{V}\cdot\mathbb{E}[V^\star(t)|Q(t)].\label{fml:ineq6}
    \end{align}
    The inequality \eqref{fml:ineq6} is because the optimal solution of $\mathbf{P1}$ may not be the optimal solution of $\mathbf{P2}$. Then we take expectation on both sides and use the law of iterated expectation
    \begin{equation}
        \mathbb{E}\left[L(Q(t+1))\right]-\mathbb{E}\left[L(Q(t))\right]-\mathcal{V}\cdot\mathbb{E}[V(t)]\leq \mathcal{B}+\mathbb{E}[Q(t)]\mathbb{E}[E^\star(t)-\bar{E}]-\mathcal{V}\cdot\mathbb{E}[V^\star(t)].\notag
    \end{equation}
    By summing over $t\in\{0,1,\cdots,T_1\}$, we obtain that
    \begin{equation}
         \mathbb{E}\left[L(Q(T))\right]-\mathcal{V}\cdot\sum_{t=0}^{T-1}\mathbb{E}[V(t)]\leq \mathcal{B}\cdot T+\sum_{t=0}^{T-1}\mathbb{E}[Q(t)]\mathbb{E}[E^\star(t)-\bar{E}]-\mathcal{V}\cdot\sum_{t=0}^{T-1}\mathbb{E}[V^\star(t)].\notag
    \end{equation}
    Divide $T$ and take limit as $T$ approaches infinity, then we have
    \begin{align}
        \lim_{T\rightarrow\infty}\frac{1}{T}\sum_{t=0}^{T-1}\mathbb{E}\left[V^\star(t)-V(t)\right]\leq\frac{\mathcal{B}}{\mathcal{V}}+\frac{1}{\mathcal{V}}\lim_{T\rightarrow\infty}\frac{1}{T}\sum_{t=0}^{T-1}\mathbb{E}[Q(t)]\mathbb{E}[E^\star(t)-\bar{E}]- \lim_{T\rightarrow\infty}\frac{\mathbb{E}\left[L(Q(T))\right]}{T}\leq\frac{\mathcal{B}}{\mathcal{V}}.\notag
    \end{align}
    Therefore, the inequality \eqref{fml:ineq4} is proved.

    Regarding inequality \eqref{fml:ineq5}, we start from \eqref{fml:ineq7} as copied below:
    $$\mathbb{E}\left[Q(T)^2\right]\leq 2T\mathcal{B}+2T\mathcal{V}(V_{max}-V_{min})-2\varepsilon\cdot\sum_{t=0}^{T-1}\mathbb{E}\left[Q(t)\right].$$
    Then we rearrange the above inequality as
    \begin{equation}
        \sum_{t=0}^{T-1}\mathbb{E}\left[Q(t)\right]\leq\frac{T\mathcal{B}}{\varepsilon}+\frac{T\mathcal{V}}{\varepsilon}(V_{max}-V_{min})-\frac{\mathbb{E}\left[Q(T)^2\right]}{2\varepsilon}\leq\frac{T\mathcal{B}}{\varepsilon}+\frac{T\mathcal{V}}{\varepsilon}(V_{max}-V_{min}).\notag
    \end{equation}
    As a result, we have
    \begin{equation}
        \lim_{T\rightarrow\infty}\frac{1}{T}\sum_{t=0}^{T-1}\mathbb{E}\left[Q(t)\right]\leq\frac{\mathcal{B}+\mathcal{V}(V_{max}-V_{min})}{\varepsilon}.\notag
    \end{equation}
    The inequality \eqref{fml:ineq5} is proved.
\end{proof}

As for the time complexity, $\mathbf{P2}$ is NP-hard. The proof can be established by setting infinite service capacity and delay tolerance in a network with only one RSU. Then $\mathbf{P2}$ can be reduced to 0-1 knapsack problem, where the latter is a well-known NP-hard problem. To tackle this challenge, we apply a heuristic algorithm in the subsection~\ref{sec:heu}.

\subsection{Model Linearization}
Notice that constraint $C_2$ contains nonlinear parts with the decision variable $y_{ijk}(t)$ in the denominator. In this subsection, some linearization tricks are applied to transform the previous optimization problem $\mathbf{P2}$ into an integer linear programming model, and then some mature solvers can be used to provide the optimal solution. 

Given prediction or observation at beginning of each time slot, the expectation in $C_2$ can be estimated via Monte Carlo method. Thus we focus on the linearization of $D_j$, and more specifically, on $L_i$ and $L_0$. Here we deal with $\mathbf{P2}$ which is oriented towards the current time slot $t$, so $t$ is omitted for simplification in this subsection.

We view $L_i$ as a decision variable, which is calculated as
$$L_i=\frac{1}{\mu_i-\sum_{j\in\mathcal{J}_i}\sum_{k\in\mathcal{K}}y_{ijk}}.$$
This is equal to the constraints below:
\begin{align}
    &L_i>0,\forall i\in\mathcal{I};\\
    &\mu_i\cdot L_i-\sum_{j\in\mathcal{J}_i}\sum_{k\in\mathcal{K}}y_{ijk}\cdot L_i=1,\forall i\in\mathcal{I}.\label{fml:con1}
\end{align}
There is a product of two decision variables in \eqref{fml:con1}, where $y_{ijk}$ is an integer variable and $L_i$ is a continuous variable. Notice $C_5$ limits the upper bound of $y_{ijk}$ and thus we have $y_{ijk}\in\{0,1,\cdots,d_{jk}\}$. Then $y_{ijk}$ is replaced by 
\begin{equation}
    y_{ijk}=\sum_{\theta=0}^{d_{jk}}\theta\cdot\chi_{ijk},
\end{equation}
where $\chi_{ijk}$ is a binary variable. Then \eqref{fml:con1} becomes
\begin{equation}
\label{fml:con2}
\mu_i\cdot L_i-\sum_{j\in\mathcal{J}_i}\sum_{k\in\mathcal{K}}\sum_{\theta=0}^{d_{jk}}\theta(\chi_{ijk}\cdot L_i)=1,\forall i\in\mathcal{I}.
\end{equation}
We rewrite \eqref{fml:con2} in terms of $\xi_{ijk}$ where
\begin{equation}
\label{fml:def1}
\xi_{ijk}=\chi_{ijk}\cdot L_i=
\begin{cases}
L_i,  &\text{if $\chi_{ijk}=1$;}  \\
0,   &\text{otherwise.}
\end{cases}
\end{equation}
And constraints for $\xi_{ijk}$:
\begin{align}
\label{fml:con3}
    &\xi_{ijk}\leq L_i, \forall i\in\mathcal{I},j\in\mathcal{J},k\in\mathcal{K};\\
    &\xi_{ijk}\leq M\cdot\chi_{ijk}, \forall i\in\mathcal{I},j\in\mathcal{J},k\in\mathcal{K};\\
    &\xi_{ijk}\geq M\cdot(\chi_{ijk}-1)+Li, \forall i\in\mathcal{I},j\in\mathcal{J},k\in\mathcal{K};\\
\label{fml:con4}
    &\xi_{ijk}\geq 0, \forall i\in\mathcal{I},j\in\mathcal{J},k\in\mathcal{K}.
\end{align}
The aforementioned constraints \eqref{fml:con3}$\sim$\eqref{fml:con4} are equivalent to \eqref{fml:def1} by examining the table below:
\begin{table}[htbp]
 \caption{Possible Combination of Constraints \eqref{fml:con3}$\sim$\eqref{fml:con4}}
  \centering
  \begin{tabular}{ccccc}
    \toprule
    $\chi_{ijk}$ & $L_i$ & $\chi_{ijk}\cdot L_i$ & \textbf{Constraints} & \textbf{Imply} \\
    \midrule
    $0$ & $L_i$ & $0$ & $\xi_{ijk}\leq L_i$ & $\xi_{ijk}=0$\\
    & & & $\xi_{ijk}\leq 0$ & \\
    & & & $\xi_{ijk}\geq -M+Li$ & \\
    & & & $\xi_{ijk}\geq 0$ & \\
    \midrule
    $1$ & $L_i$ & $L_i$ & $\xi_{ijk}\leq L_i$ & $\xi_{ijk}=L_i$\\
    & & & $\xi_{ijk}\leq M$ & \\
    & & & $\xi_{ijk}\geq Li$ & \\
    & & & $\xi_{ijk}\geq 0$ & \\
    \bottomrule
  \end{tabular}
\end{table}

The linearization of $L_0$ follows a similar trick as mentioned above. Here we omit it for brevity.

\subsection{Heuristic Algorithm}
\label{sec:heu}
As aforementioned, solving $(\mathbf{P2})$ becomes increasingly challenging within a reasonable time as network size grows. Therefore we apply Binary Quantum-behaved Particle Swarm Optimization (BQPSO) to expedite the decision-making process.
As a famous variant of particle swarm optimization (PSO), the quantum-behaved PSO (QPSO) keeps the philosophy of PSO and applies statistical simulation for approximating the global optimal solution in quantum space. To cope with discrete space, the binary QPSO (BQPSO) algorithm is introduced via mapping strategy and discretization mechanism, which has distinct superiority in convergence and performance~\cite{li2023improved}. 
Given the proposed optimization problem $\mathbf{P2}$, we first transform it into an unconstrained model via penalty method as follows:
\begin{equation}
	\label{fml:fitness}
	F(\mathbf{X},\mathbf{Y})=Obj(\mathbf{X},\mathbf{Y})+\gamma\cdot Pen(\mathbf{X},\mathbf{Y}),
\end{equation}
where $Obj$ is the objective function in $\mathbf{P2}$, $Pen$ is the penalty function that is converted from constraints $C_2-C_5$, and $\gamma$ is the penalty factor. Here we deal with the caching allocation for each time slot, thus $t$ is omitted from variables for simplification. Note that the decision variables $\mathbf{X},\mathbf{Y}$ may generate an infeasible allocation during BQPSO processing. To improve searching efficiency, we set $\mathbf{X}$ as the pivot variable and then $\mathbf{Y}$ can be allocated equally to the available RSUs:
\begin{equation}
    \label{fml:x2y}
    y_{ijk}=\frac{d_{jk}}{\sum_{i\in\mathcal{I}_j}x_{ik}}
\end{equation}

Let $\tilde{N}$ be the number of particles, $\tilde{D}$ be the dimension of $\mathbf{X}$ which equals to $I\times K$ in this paper, $\mathbf{\tilde{X}}_n$ be the location of $n^{\text{th}}$ particle, $\mathbf{\tilde{P}}_n$ be the historical individual best location, and $\mathbf{\tilde{P}}^*$ be the recorded global best position. In particular, $\mathbf{\tilde{X}}_n$ is the vectorization of $\mathbf{X}$. In BQPSO, the continuous local attractor $\mathbf{\tilde{L}}^\aleph$ at iteration $\tilde{t}$ is calculated as
$$
	\tilde{L}^\aleph_{nd}(\tilde{t})=\phi\cdot\tilde{P}_{nd}(\tilde{t})+(1-\phi)\cdot \tilde{P}_{d}^*(\tilde{t}),
$$
and then $\tilde{L}^\aleph_{nd}(\tilde{t})$ is mapped to binary local attractor $\tilde{L}_{nd}(\tilde{t})$ via
\begin{equation}
\label{fml:attractor_update}
    \tilde{L}_{nd}(\tilde{t})=
    \begin{cases}
        1,  &\text{if $\psi<(1+e^{-\tilde{L}^\aleph_{nd}(\tilde{t})})^{-1}$;}  \\
        0,   &\text{otherwise.}
    \end{cases}
\end{equation}
The continuous position of particle $n$ at iteration $\tilde{t}$ is updated by
\begin{equation}
	\label{fml:continuous_update}
	\tilde{X}^\aleph_{nd}(\tilde{t}\!+\!1)\!=\!\tilde{L}^\aleph_{nd}(\tilde{t})\!\pm\!\eta\left|\frac{1}{\tilde{N}}\sum_{n=1}^{\tilde{N}}\!\tilde{P}_{nd}(\tilde{t})\!-\!\tilde{X}^\aleph_{nd}(\tilde{t})\right|\log\!\left(\frac{1}{G}\right)
\end{equation}
and the discrete position is updated through
\begin{equation}
    \label{fml:discrete_update}
    \tilde{X}_{nd}(\tilde{t})=
    \begin{cases}
        1,  &\text{if $\tilde{X}^\aleph_{nd}(\tilde{t})\geq\cfrac{1}{\tilde{D}}\sum_{d=1}^{\tilde{D}}\tilde{X}^\aleph_{nd}(\tilde{t})$;}  \\
        0,   &\text{otherwise.}
    \end{cases}
\end{equation}
where $\phi$, $\psi$ and $G$ are generated based on the uniform distribution $[0,1]$. $\eta$ is the contraction expansion factor. In this paper, an adaptive $\eta$ is employed where $$\eta(t)=\omega_1+\frac{(\omega_2-\omega_1)(\tilde{T}-\tilde{t})}{\tilde{T}}.$$ Furthermore, the $+$ or $-$ in \eqref{fml:continuous_update} is selected randomly. 
The details of BQPSO framework is depicted in Alg.\ref{alg:GQPSO}. Obviously, the time complexity is $O(\tilde{N}\cdot\tilde{T})$, which depends on population size $\tilde{N}$ and maximum iteration number $\tilde{T}$.
\begin{algorithm}[!htbp]
	\caption{BQPSO-based Decision Algorithm (BQPSO-DA)}
	\label{alg:GQPSO}
	\KwIn{Variables in Alg.\ref{alg:online}; $\omega_1$, $\omega_2$}
	Initialize $\mathbf{\tilde{X}}$ and evaluate $F$ via \eqref{fml:fitness}\;
	Record current $\mathbf{\tilde{P}}$ and $\mathbf{\tilde{P}}^*$\;
	\While{not exceed maximum number of iterations $\tilde{T}$}{
		\ForEach{individual particles $n$}{
			Update $\mathbf{\tilde{X}}^\aleph_n$ through \eqref{fml:continuous_update}\;
			Discrete $\mathbf{\tilde{X}}_n$ through \eqref{fml:discrete_update}\;
            Calculate $\mathbf{\tilde{L}}_{n}$ through \eqref{fml:attractor_update}\;
            \If{Hamming distance between $\mathbf{\tilde{X}}_n$ and $\mathbf{\tilde{L}}_{n}$ $\geq\tilde{N}/2$}{
                Perform single-point crossover operation\;}
			Evaluate $F$ via \eqref{fml:fitness} and update $\mathbf{\tilde{P}}_n$\;
		}	
		Update $\mathbf{\tilde{P}}^*$ and iteration counter\;
	}
	Convert $\mathbf{\tilde{P}}^*$ into $\mathbf{X}$ and determine $\mathbf{Y}$ via \eqref{fml:x2y}\;
\end{algorithm}
\section{Numerical Investigations}
\label{sec:simulation}
\subsection{Simulation Setup}

To evaluate the performance of proposed algorithms, we carry out simulations within an area of 800m$\times$500m consisting of $1$ horizontal and $2$ vertical bidirectional streets, which is divided into $16$ regions equally. The simulation incorporates $6$ RSUs, with each RSU responsible for serving between $5$ to $8$ regions.
Notice that RSUs are densely deployment and thus there are overlaps in service areas. We designate that each region is associated with $5$ distinct SD instances to reflect region status. The requested SD size ranges from $1$Mb to $10$Mb. The lifespan of each SD varies from $1$s to $100$s. For the sake of simplicity, the space validity of each SD is fixed to a 100m radius. The caching capacity of each RSU follows a uniform distribution of $[100,500]$Mb. Additionally, the delay tolerance for each region is set to $500$ms and the caching power coefficient is $2.5\times10^{-9}$W/b. 
As for the communication process, the downlink channel gain is calculated by the following path-loss model $P_L(d)=\bar{P}_L(d_0)+10\gamma\log(d/d_0)+X_\sigma$, where $d$ is the average distance between region and RSU, $d_0=1$km is the reference distance and the mean path-loss at reference point $\bar{P}_L(d_0)$ is set to 28dB, $\gamma=2$ is the path-loss exponent in free space, shadowing effect $X_\sigma$ is modeled as a log-normal distribution with 8dB standard deviation and zero mean. The transmission power is set to $30$dBm and the bandwidth is $20$MHz. Additionally, the noise power spectral density is $-174$dBm/Hz. 
Regarding the GQPSO in Alg.\ref{alg:GQPSO}, we set both the population size and the maximum number of iterations as $100$, $\omega_1$ and $\omega_2$ for calculating contraction factor as $0.5$ and $1$ respectively.

We compare the proposed schemes against some well-known algorithms such as Greedy Caching and Random Caching, which have been introduced in \cite{liu2021cooperative} and \cite{chen2017probabilistic} respectively. One proposed method is Alg.\ref{alg:online} solving linearized $\mathbf{P2}$ via branch-and-bound method and we use term OCDA to represent it hereafter.
The other proposed approach is solving $\mathbf{P2}$ through BQPSO-DA. Regarding to Greedy Caching, SD will be cached and updated based on the caching value $V(t)$ and RSU capacity. For each requested SD, Random Caching chooses positions randomly among RSUs which cover the relevant regions and updates contents in a FIFO manner. Drawing from congestion definition as reported in  \cite{China2020congestion}, we conduct our simulations under four distinct traffic scenarios: no congestion, light congestion, moderate congestion, and heavy congestion, which correspond to vehicle counts of 6, 8, 10, and 12 per region, respectively. All results in this section are compared over $1800$ time slots and $\tau=1$s.

\subsection{Performance Comparison}

\begin{figure}[htbp]
\centering
\begin{minipage}{.33\linewidth}
\centering
\includegraphics[width=\linewidth,trim=0 18em 0 0]{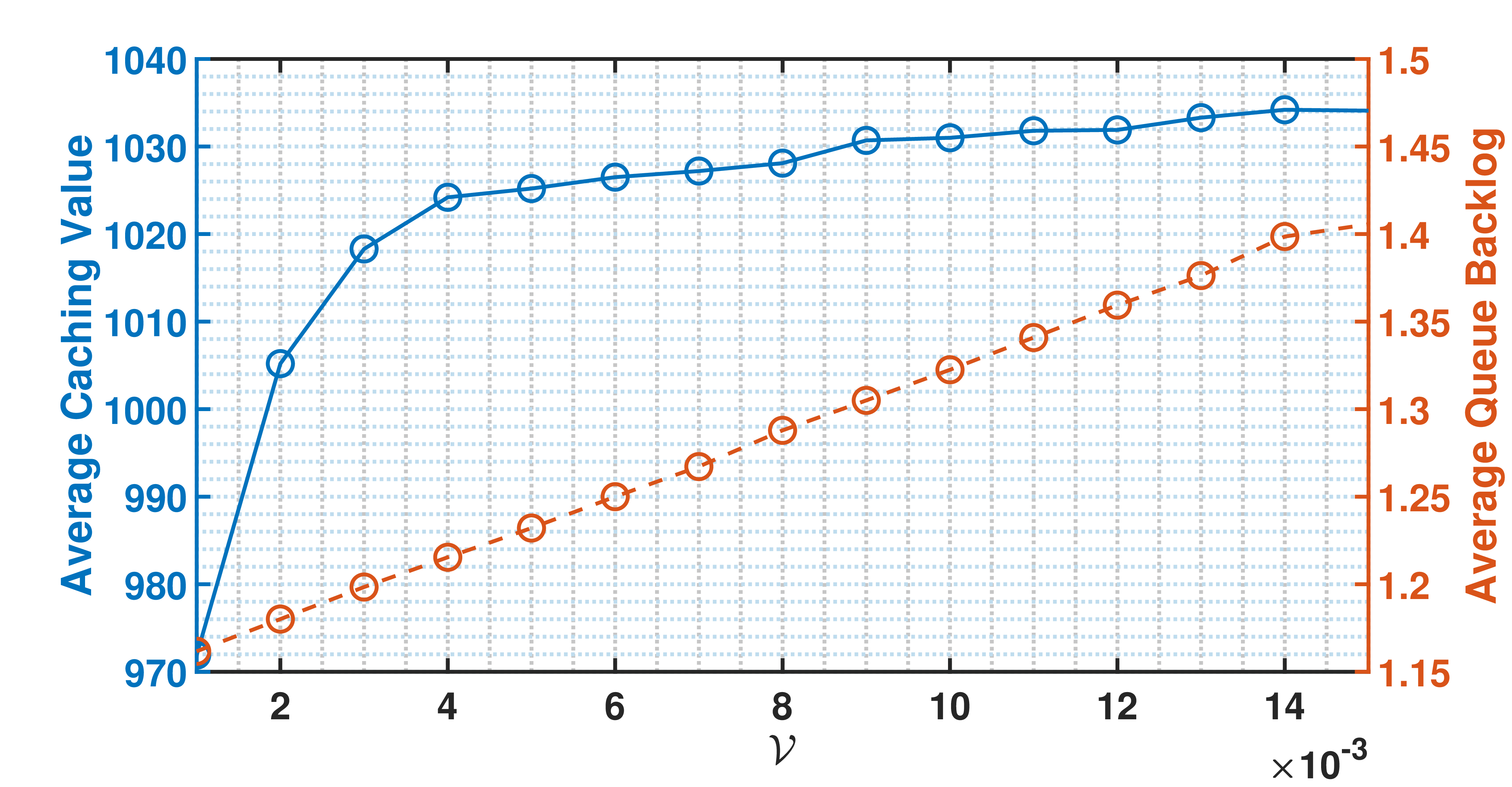}
\caption{Impact of $\mathcal{V}$}
\label{fig:V}
\end{minipage}
\begin{minipage}{.33\linewidth}
\centering
\includegraphics[width=\linewidth,trim=0 18em 0 0]{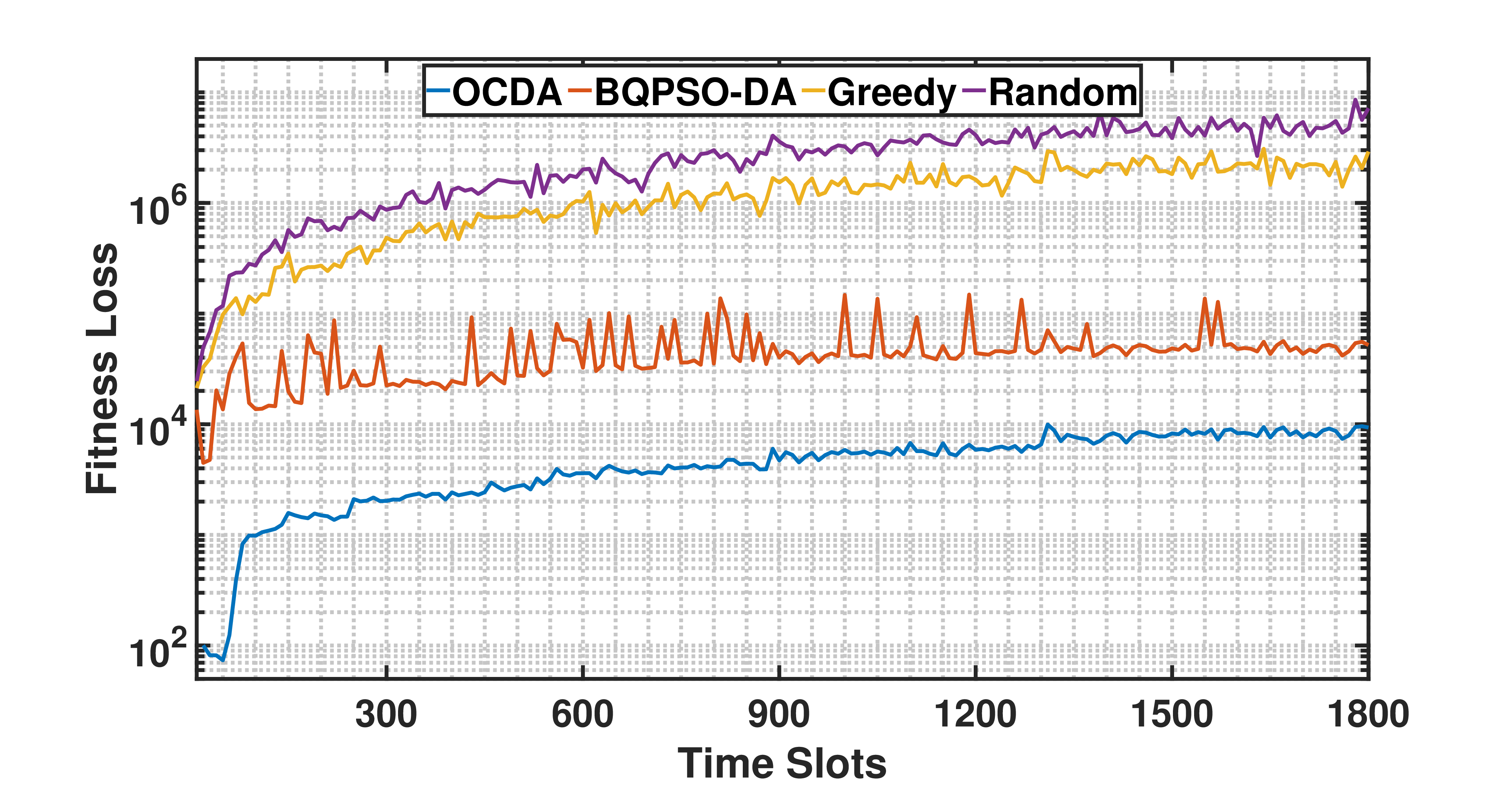}
\caption{Fitness Loss $F(\mathbf{X},\mathbf{Y})$}
\label{fig:fitness}
\end{minipage}
\begin{minipage}{.33\linewidth}
\centering
\includegraphics[width=\linewidth,trim=0 18em 0 0]{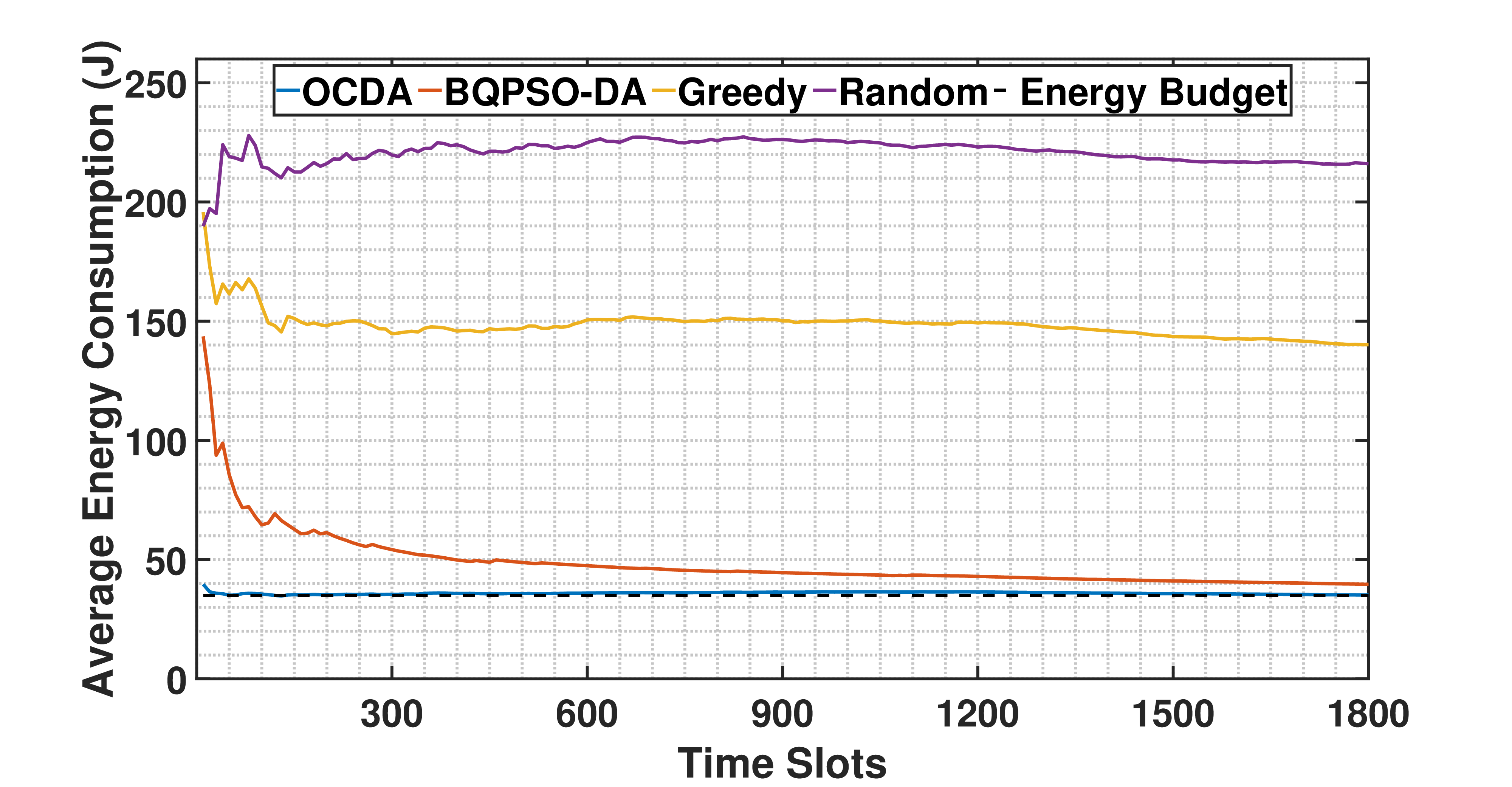}
\caption{Time Average Energy}
\label{fig:energy_mean}
\end{minipage}

\begin{minipage}{.33\linewidth}
\centering
\includegraphics[width=\linewidth,trim=0 18em 0 0]{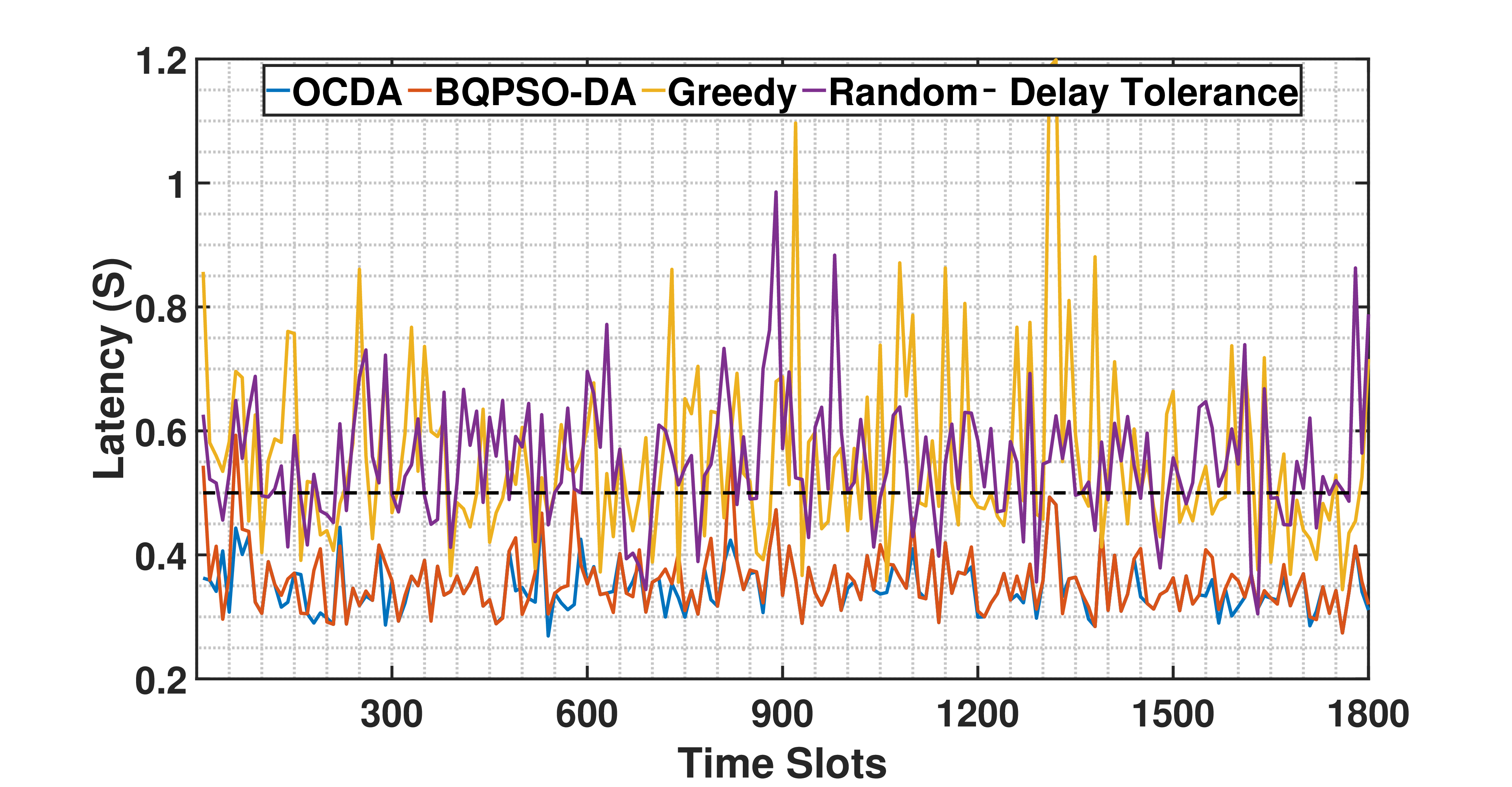}
\caption{Maximum Delay in Regions}
\label{fig:delay}
\end{minipage}
\begin{minipage}{.33\linewidth}
\centering
\includegraphics[width=\linewidth,trim=0 18em 0 0]{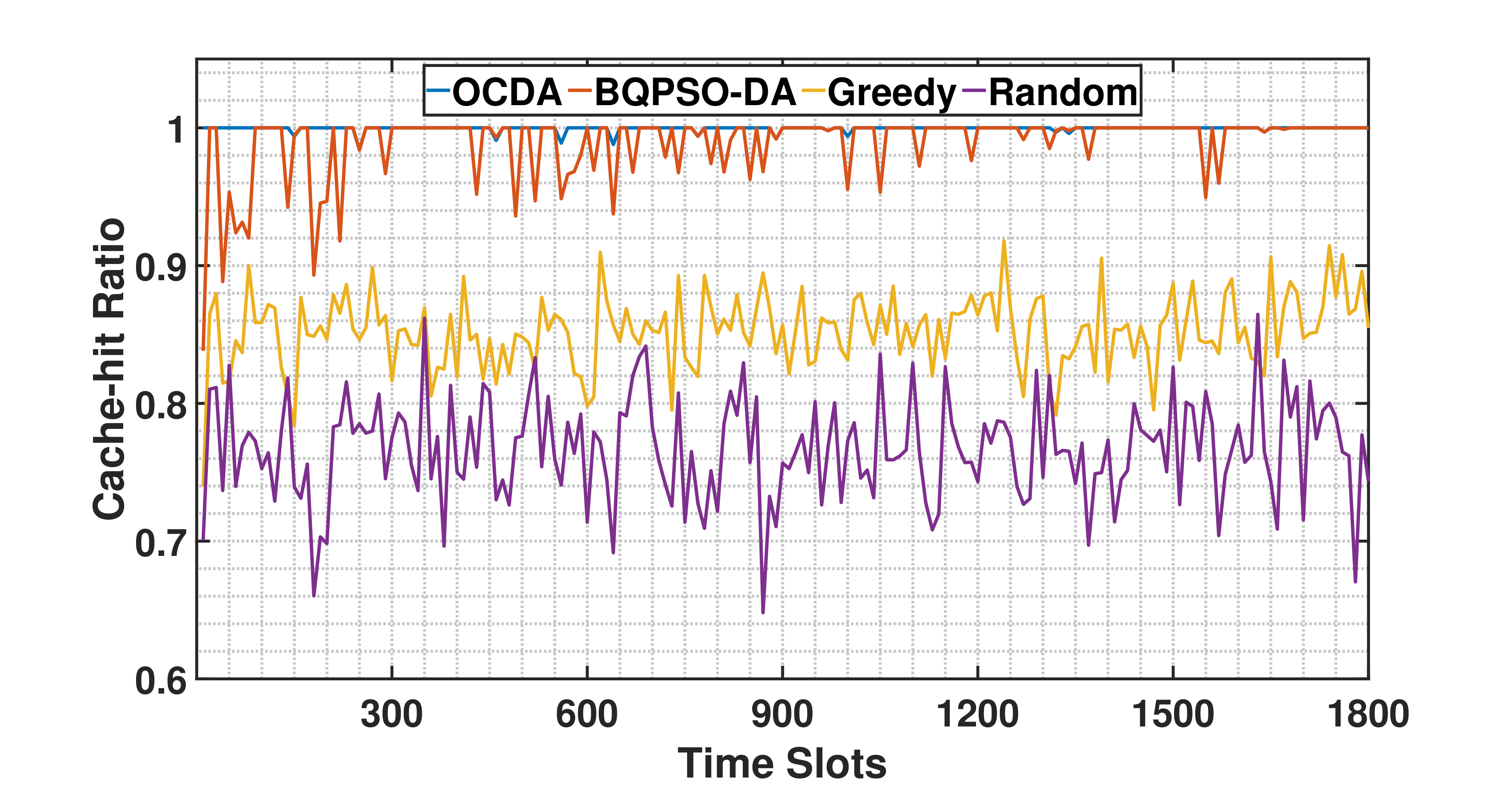}
\caption{Cache-hit Ratio}
\label{fig:ratio}
\end{minipage}
\begin{minipage}{.33\linewidth}
\centering
\includegraphics[width=\linewidth,trim=0 18em 0 0]{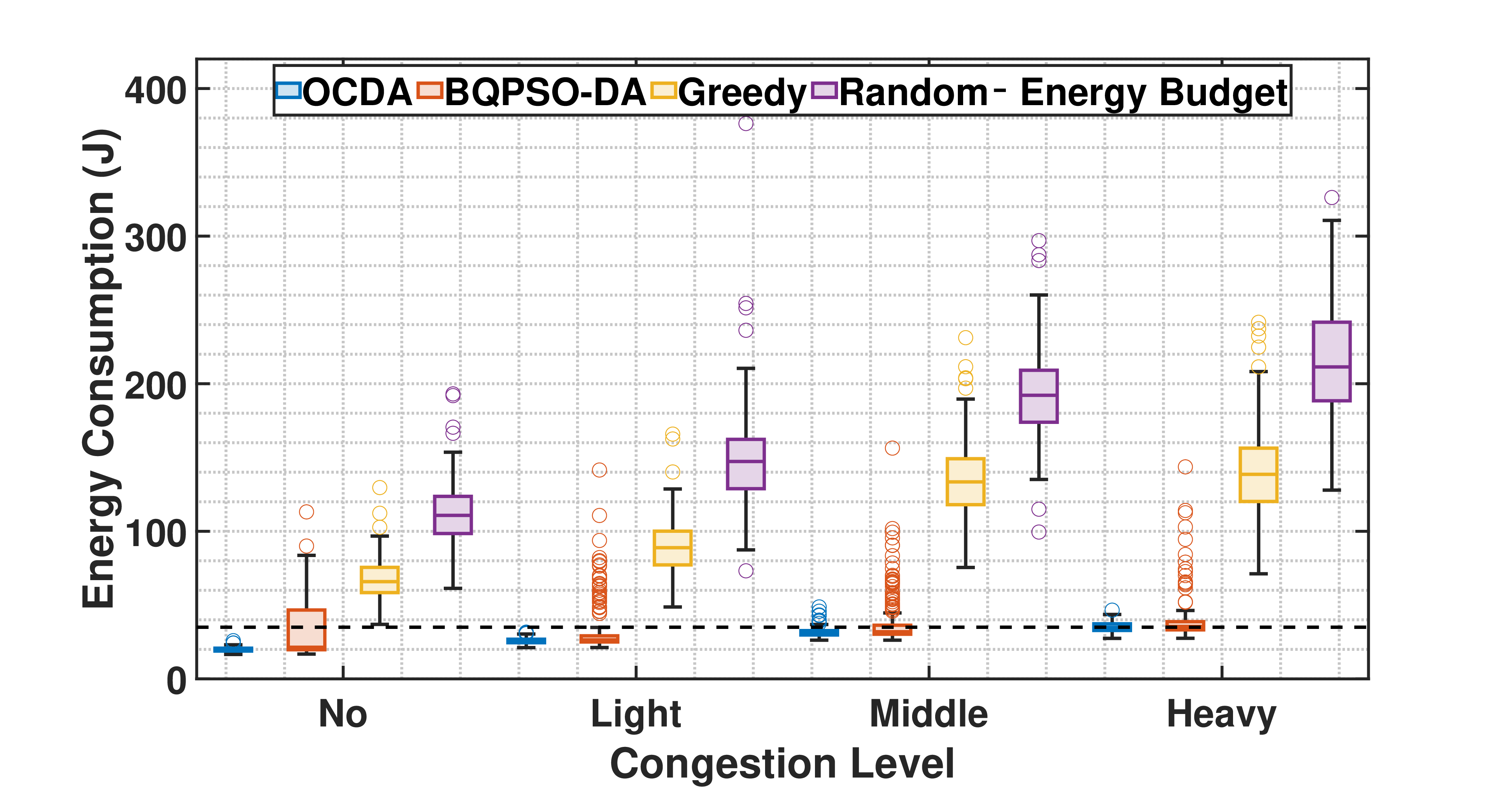}
\caption{Energy in 4 Cases}
\label{fig:energy_cases}
\end{minipage}

\begin{minipage}{.33\linewidth}
\centering
\includegraphics[width=\linewidth,trim=0 18em 0 0]{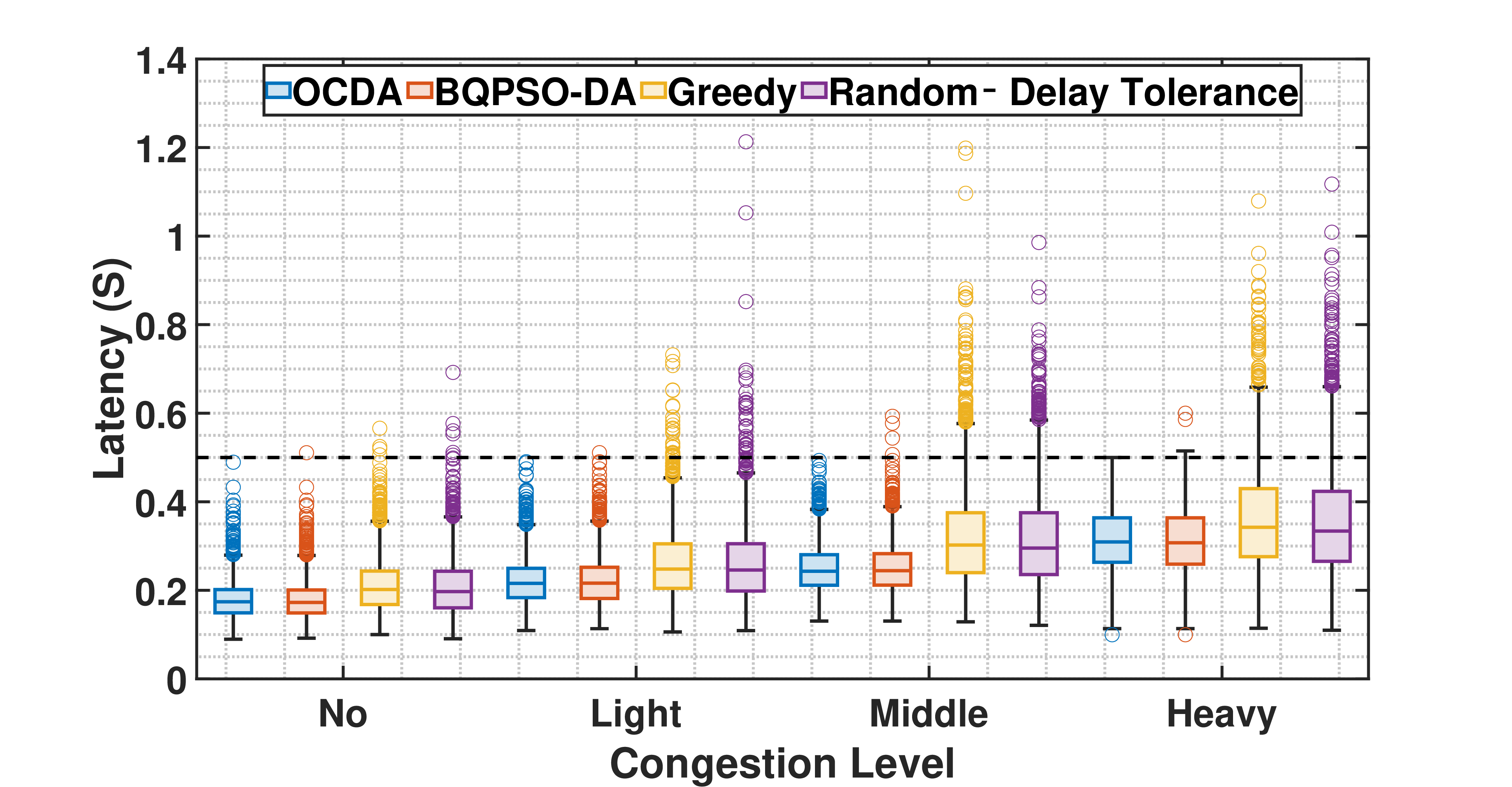}
\caption{Latency in 4 Cases}
\label{fig:delay_cases}
\end{minipage}
\begin{minipage}{.33\linewidth}
\centering
\includegraphics[width=\linewidth,trim=0 18em 0 0]{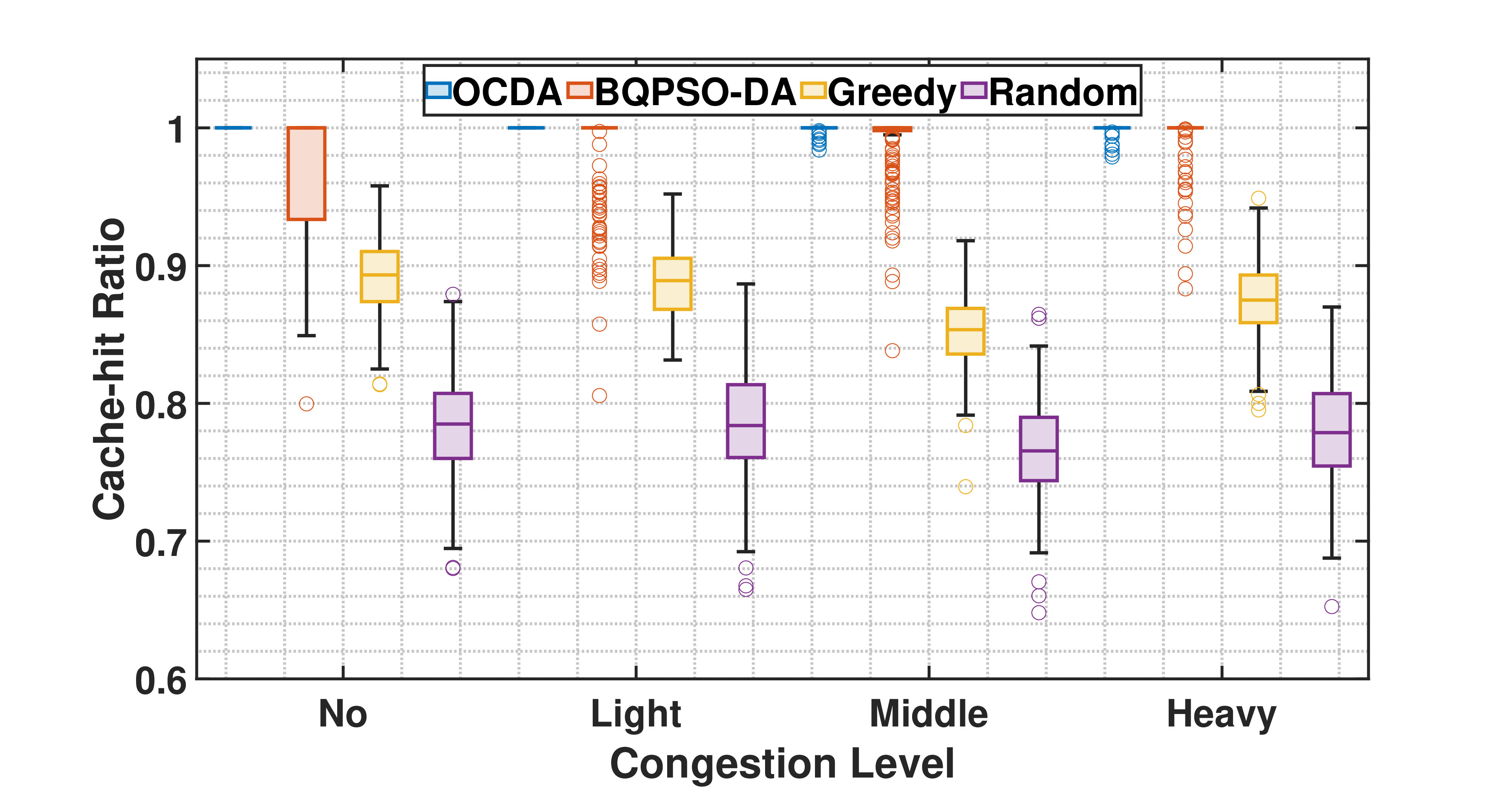}
\caption{Cache-hit Ratio in 4 Cases}
\label{fig:ratio_cases}
\end{minipage}
\begin{minipage}{.33\linewidth}
\centering
\includegraphics[width=\linewidth,trim=0 18em 0 0]{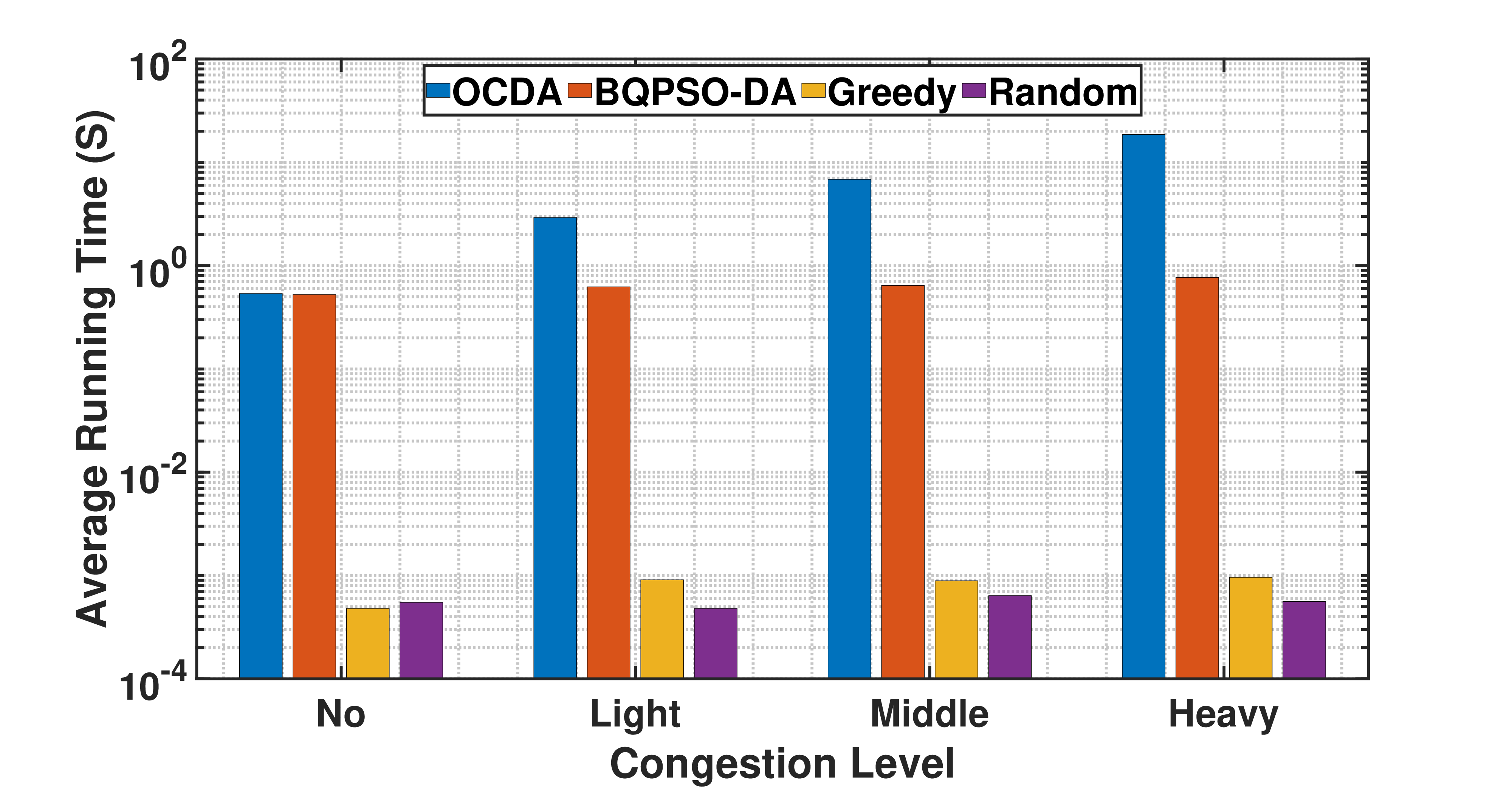}
\caption{Running Time in 4 Cases}
\label{fig:time}
\end{minipage}
\end{figure}

We first validate the theoretical performance guarantee of OCDA in Fig.\ref{fig:V}. With the energy budget $\bar{E}$ consistently set at $35$J, we increase $\mathcal{V}$ from $10^{-3}$ to $1.5\times10^{-2}$. Initially, there is a significant surge in the caching value, which subsequently plateaus at a near-optimal threshold. This trend indicates an inversely proportional reduction in the performance gap between OCDA and the optimal solution as $\mathcal{V}$ grows. Meanwhile, the time-averaged queue backlog is observed to increase proportionally to $\mathcal{V}$. Given that the enhancement in caching value becomes marginal beyond $\mathcal{V}=4\times10^{-3}$, we adopt this value for following analyses.

Then we delve into the performance of different schemes under middle congestion scenario. Fig.\ref{fig:fitness} depicts the fitness loss $F(\mathbf{X},\mathbf{Y})$ as defined by \eqref{fml:fitness}, in which OCDA provides most favorable outcomes, signifying its superior capability to maintain a proper state compared to other methods. BQPSO-DA also delivers a competitive solution compared to OCDA. In Fig.\ref{fig:energy_mean}, the time-averaged energy expenditure of BQPSO-DA nearly aligns with the long-term energy budget (represented by the black dot line) as OCDA performs. Furthermore,  Fig.\ref{fig:delay} demonstrates that BQPSO-DA remains beneath the latency threshold in the majority of time slots, while OCDA can consistently satisfy the delay limitation. 
Contrasting with OCDA and BQPSO-DA, both the Greedy strategy and Random Caching algorithm exhibit diminished cache-hit ratios in Fig.\ref{fig:ratio}. This shortfall results in a higher reliance on BS retrieval for SD demands, consequently translating to more energy consumption in Fig.\ref{fig:energy_mean} and excessive delay in Fig.\ref{fig:delay}. Benefiting from the prior knowledge of SD caching values $V(t)$, the Greedy approach outperforms Random Caching.  

Furthermore, Fig.\ref{fig:energy_cases}, \ref{fig:delay_cases}, and \ref{fig:ratio_cases} show different performance distributions under varying congestion levels, spanning from no congestion to heavy congestion. Overall, the performance generally deteriorates as congestion intensifies. Among four evaluated algorithms, OCDA demonstrates superior performance over all metrics. Particularly in high congestion scenarios as shown in Fig.\ref{fig:energy_cases}, the distribution of OCDA energy consumption is in close proximity to the budget line, which contributes to the virtual queue stability. Moreover, OCDA consistently meets the lantecy requirements and sustains an almost perfect cache-hit ratio across all congestion levels, as showcased in Fig.\ref{fig:delay_cases} and \ref{fig:ratio_cases}, respectively. 
BQPSO-DA also provides acceptable performance in light to heavy congestion cases, yet occasional outliers are observed, indicating sporadic instances of high energy consumption and response delay that impact the algorithmic robustness. Additionally, these figures reveal that the Greedy and Random Caching are suitable in scenarios with less stringent limitations and lighter congestion. 

Finally, an assessment of time complexity is shown in Fig.\ref{fig:time}. The running time for OCDA varies from 0.5s to more than 20s. This variance implies the necessity for OCDA to perform caching allocation well in advance and the requirement of an enhanced foresightful prediction algorithm. In contrast, BQPSO-DA can generate solutions within one time slot, which is more amendable to a slot-by-slot decision-making paradigm than OCDA. Greedy and Random Caching are significantly faster than two proposed algorithms, with computation time requiring less than 10ms.   
\section{Conclusions}
\label{sec:conclusions}

Based on the caching value model that encapsulates the temporal and spacial features of SD, we have introduced a stochastic programming framework addressing both caching placement and request allocation in this paper. Then we have decomposed the long-term model via the Lyapunov optimization. To refine efficiency, we have provided two tailored methods: the OCDA, a linearized online method optimized for overall performance; and the BQPSO-DA, a heuristic algorithm prioritized for swift execution. Numerical investigations reveal that both proposed schemes provide competitive solutions. Furthering this work, we plan to explore AI-driven algorithms to augment proposed methods with learning capabilities and reduce computational overhead.

\section*{Acknowledgments}
This work was sponsored in part by the National Key R\&D Program of China under Grant 2023YFE0208800, in part by Shandong Provincial Natural Science Foundation, China under Grant ZR2023QF084.

\bibliographystyle{unsrt}  
\bibliography{references}  

\end{document}